 \documentclass[journal, 11pt]{IEEEtran}
\onecolumn
\usepackage{mathrsfs}
\usepackage{tikz}
\usepackage[utf8]{inputenc}
\usepackage{amsfonts} 
\usepackage{epstopdf}
\usepackage{graphicx}
\input{xypic}
\usepackage{bbm}
\usepackage{enumerate}
\usepackage{epsf,epsfig,verbatim,amssymb,amsmath,array,cite,amsthm,subfigure,multicol,multirow}  
\usepackage{psfrag,bm,xspace}
\usepackage{slashbox}
\usepackage{hhline}

\renewcommand{\baselinestretch}{1.27}
\newcommand{\beq}{\begin{equation}}
\newcommand{\eeq}{\end{equation}}

\newtheorem{lem}{Lemma}

\newtheorem{theorem}{Theorem}

\newtheorem{definition}{Definition}

\theoremstyle{definition}
\newtheorem{example}{Example}

\newtheorem*{prob*}{Problem}

\theoremstyle{remark}
\newtheorem{remark}{Remark}

\allowdisplaybreaks 

\global\long\def\RR{\mathbb{R}}

\global\long\def\ZZ{\mathbb{Z}}
\global\long\def\EE{\mathbb{E}}

\global\long\def\FF{\mathbb{F}}
\global\long\def\11{\mathbbm{1}}

\allowdisplaybreaks

\begin{document}
\title{How to Compute Modulo Prime-Power Sums}

\author{
 \IEEEauthorblockN{Mohsen Heidari}
  \IEEEauthorblockA{EECS Department\\University of Michigan\\ Ann Arbor,USA \\
    Email: mohsenhd@umich.edu} 

  \and
  \IEEEauthorblockN{S. Sandeep Pradhan}
  \IEEEauthorblockA{EECS Department\\University of Michigan\\ Ann Arbor,USA \\
    Email: pradhanv@umich.edu}
}

%
%


\IEEEoverridecommandlockouts


\maketitle
\begin{abstract}
 The problem of computing modulo prime-power sums is investigated in distributed source coding as well as computation over Multiple-Access Channel (MAC). We build upon group codes and present a new class of codes called Quasi Group Codes (QGC). A QGC is a subset of a group code. These codes are not closed under the group addition.  We investigate some properties of QGC's, and  provide a packing and a covering bound. Next, we use these bounds to derived achievable rates for distributed source coding as well as computation over MAC. We show that strict improvements over the previously known schemes can be obtained using QGC's.
\end{abstract}

\section{Introduction}
\IEEEPARstart{E}{ver} since the seminal paper by  Korner and Marton  in 1979, structured codes played a 
key role in the study of  asymptotic performance of multi-terminal communications  \cite{korner-marton}-\cite{Arun_comp_over_MAC_ISIT13}. 
In all of these works, algebraic structure of the codes is exploited to derive new bounds on the asymptotic 
performance limits of communication.
These bounds are strictly better than those derived using unstructured codes.  
Most of these works concentrate on linear codes built on finite fields. 
Despite the aforementioned benefits, the algebraic structure imposed by linear codes has certain restrictions.  
Finite fields exist only when the alphabet size is a prime power. Even when the existence is not an issue, in 
certain problems, weaker algebraic structures such as groups have better properties \cite{Loeliger-91}.  
Group codes are a type of structured codes that are closed under the group operation. These codes have 
been studied in \cite{Loeliger-91}- \cite{Aria_group_codes} for point-to-point (PtP) 
communication problems. Under specific constraints in multi-terminal settings, compared to linear codes, the structure 
of group codes matches better with that of the channel or source. This results in achieving lower transmission rates in 
certain distributed source coding problems  \cite{Aria_dist_source}  and 
higher transmission rates for certain broadcast channels  \cite{Aron-BC-ISIT13}.

When the underlying group is not a field, there are non-trivial subgroups. Since group codes are closed under 
the group addition, these subgroups put a penalty on the transmission rates.  Based on this observation, 
in our earlier attempt,  we introduced a  class of structured codes called transversal group codes \cite{transversal}. 
These codes are built over cyclic groups. In contrast to group codes, they are not closed under the group 
addition. This allows the transversal group codes to compensate for the penalty put by  subgroups and achieve higher/lower 
transmission rates in channel/source coding problems. In particular, these codes extend the asymptotic rate 
region achievable in distributed source coding as well as computation over MAC. 

In this paper, we extend the notion of transversal group codes and introduce a new class of  codes over groups 
called Quasi Group Codes (QGC). These codes are constructed by taking subsets of group codes. 
We restrict  ourselves to cyclic groups and  provide a  construction of the subsets. We first study some basic properties of 
QGC's and derive a packing and a covering bound for such codes. These bounds indicate that the PtP 
channel capacity and optimal rate-distortion function is achievable using QGC's. Next, we use these results 
to explore the applications of QGC's in multi-terminal communication problems. We derive achievable 
rates using QGC's for certain distributed source coding and computation over MAC problems. 
We show, through some examples, that these codes give better achievable rates for both settings. 

The rest of this paper is organized as follows: Section \ref{sec: preliminaries} provides the preliminaries and notations. In Section \ref{sec: proposed scheme} we propose QGC's and investigate some of their properties. In Section  \ref{sec: dist} and Section \ref{sec: comp_over_mac}, we discuss the applications of QGC's in distributed source coding and computation over MAC, respectively. Section \ref{sec: conclusion} concludes the paper.

\section{Preliminaries} \label{sec: preliminaries}
\subsection{Notations}
We denote (i) vectors using lowercase bold letters such as $\mathbf{b}, \mathbf{u}$, (ii) matrices using uppercase bold letters such as $\mathbf{G}$, (iii) random variables using capital letters such as $X,Y$, (iv) numbers, realizations of random variables and elements of sets using lower case letters such as $a, x$. Calligraphic letters such as $\mathcal{C}$ and  $\mathcal{U}$  are used to represent sets. For shorthand, we denote the set $\{1, 2, \dots, m\}$ by $[1:m]$. 

\subsection{Definitions}
A group is a set equipped with a binary operation denoted by “$+$”. Given a prime power $p^r$, the group of integers modulo-$p^r$ is denoted by $\ZZ_{p^r}$, where the underlying set is $\{0,1,\cdots, p^r-1\}$, and the addition  is modulo-$p^r$. For $s \in \{0,1,\cdots, r\}$, define $$H_{s}=p^{s}\ZZ_{p^r}=\{0, p^{s}, 2p^{s}, \cdots,  (p^{r-s}-1)p^{s}\},$$ and $T_s=\{0,1, \cdots
 , p^s-1\}$. For example, $H_0=\ZZ_{p^r}, T_0=\{0\}$, whereas $H_r=\{0\}, T_r=\ZZ_{p^r}$. Note, $H_s$ is a subset of $\ZZ_{p^r}$ that is closed under the modulo-$p^r$ addition. Given $H_{s}$ and $T_s$, each element $a$ of $\ZZ_{p^r}$ can be represented uniquely as a sum $a=t+h$, where $h\in H_{s}$ and $t \in T_s$. We denote such $t$ by $[a]_{s}$. 
 
For any elements $a,b \in \ZZ_{p^r}$, we define the multiplication $a\cdot b$ by adding $a$ with itself $b$ times. Given a positive integer $n$,  denote $\ZZ_{p^r}^n=\bigotimes_{i=1}^n\ZZ_{p^r}$. Note $\ZZ_{p^r}^n$ is a group, whose addition is element-wise and its underlying set is $\{0,1, \dots, p^r-1\}^n$.   
\begin{definition}[Shifted Group Codes]\label{def: group codes}
A \textit{group code} $\mathcal{C}$ over $\ZZ_{p^r}$ with length $n$ is a subgroup of $\ZZ_{p^r}^n$. A \textit{shifted group code}  over $\ZZ_{p^r}$ is a translation of a group code $\mathcal{C}$ by a fixed element $\mathbf{b}\in \ZZ_{p^r}^n$.   
\end{definition}
%
Shifted group codes, in general, are defined over arbitrary groups. Sahebi, \textit{et al,} \cite{Aria_group_codes}, characterized the ensemble of all group codes over finite commutative groups. 
%
\begin{definition}[Transversal Group Codes]
Consider non-negative integers $n, k_1, k_2, \dots, k_r$.  An $(n,k_1,k_2, \dots, k_r)$-transversal group code over $\ZZ_{p^r}$ is defined as 
$$\mathcal{C}=\{ \sum_{s=1}^r \mathbf{u}_s \mathbf{G}_s+\mathbf{b}: \mathbf{u}_s \in T_s^{k_s}, s\in [1:r]\},$$ where  $\mathbf{b}\in \ZZ^n_{p^r}$ and $\mathbf{G}_s$ is a $k_s \times n$ matrix with elements in $\ZZ_{p^r}$.
\end{definition}
Suppose that the elements of $\mathbf{G}_s$ and $\mathbf{b}$ are selected randomly and uniformly over $\ZZ_{p^r}$. Then, for large enough $n$, with probability close to one, the rate of this code equals $$R=\frac{1}{n}\log_2 |\mathcal{C}| =\sum_{s=1}^r \frac{k_s}{n}\log_2|T_s|=\sum_{s=1}^r \frac{k_s}{n}\log_2p^s$$
Performance limits of transversal codes for point-to-point as well as certain multi-terminal problems are investigated in \cite{transversal}.

Consider a two user MAC whose input alphabets at each terminal is $\ZZ_{p^r}$, and its output alphabet is denoted by $\mathcal{Y}$.
\begin{definition} [Codes for computation over MAC]\label{def: code for comp over MAC}
 A $(\theta_1, \theta_2)$-code for computation over the above MAC consists of two encoding functions and one decoding function. The encoding functions are denoted by $f_i:[1:\theta_i]\rightarrow \ZZ_{p^r}^n$, for $i=1,2$, and the decoding function is a map $g: \mathcal{Y}^n \rightarrow \ZZ_{p^r}^n$.
\end{definition}

\begin{definition}[Achievable Rate] \label{def: comp over MAC achievable rate}
$(R_1,R_2)$ is said to be achievable, if for any $\epsilon >0$, there exist a $(\theta_1, \theta_2)$-code such that  
\begin{align*}
&P\{g(Y^n)\neq f_1(M_1)+f_2(M_2)\}\leq \epsilon, \quad 
R_i \leq \frac{1}{n}\log \theta_i,
\end{align*}
where $M_1$ and $M_2$ are independent random variables and $p(M_i=m_i)=\frac{1}{\theta_i}$ for all $ m_i \in [1:\theta_i], i=1,2$.
\end{definition}

\section{Quasi Group Codes}\label{sec: proposed scheme}
 A linear code over a field $\FF_p$ is defined as a subspace of $\FF_p^n$. This code can also be viewed as the image of a linear transformation from $\FF_p^k$ into $\FF_p^n$.  Similarly, a shifted group code over $\ZZ_{p^r}$ (as in Definition \ref{def: group codes}) is the image of an \textit{addition-preserving} map from $\ZZ_{p^r}^k$ into $\ZZ_{p^r}^n$. This map is denoted by $\phi(\mathbf{u})=\mathbf{u}\mathbf{G}+\mathbf{b}$, where $\mathbf{b}\in \ZZ_{p^r}^n$ and $\mathbf{G}$ is a $k\times n$ matrix whose elements are in $\ZZ_{p^r}$.

The idea to construct a  \textit{quasi group} code is to consider only a subset of a shifted group code. This can be done by restricting the domain of $\phi$ to a subset $\mathcal{U}$ of  $\ZZ_{p^r}^k$. Therefore, a QGC is defined by 
$$ \mathcal{C}=\{\mathbf{u}\mathbf{G}+\mathbf{b}: \mathbf{u}\in \mathcal{U}\}, $$
where $\mathcal{U}$ is an arbitrary subset of $\ZZ_{p^r}^k$. For a general subset $\mathcal{U}$, the codebook $\mathcal{C}$ is not necessary closed under the addition. However, it possesses certain algebraic structures. Note that it is difficult to analyze the performance of QGC for a general $\mathcal{U}$. In what follows, we present a special construction of $\mathcal{U}$ that is suitable for tractability in analyzing the performance of the code.   

Let $U$ be a random variable over $\ZZ_{p^r}$ and set $\mathcal{U}=A_{\epsilon}^{(k)}(U)$. In this case, by changing the PMF of $U$, one can create different sets $\mathcal{U}$. For example, if $U$ is uniform over $\ZZ_{p^r}$, then $\mathcal{U}=\ZZ_{p^r}^k$, and $\mathcal{C}$ will become a shifted group code. 

Next, we provide a more general construction of $\mathcal{U}$. Fix $m$, and consider positive integers $k_i, i\in [1:m]$.  For each $i$, let $\mathbf{G}_i$ be a $k_i \times n$ matrix with elements in $\ZZ_{p^r}$. Suppose $U_1, U_2, \cdots, U_m$ are independent random variables over $ \ZZ_{p^r}$. As a codebook define 
\begin{align}\label{eq: codebook generation}
\mathcal{C}=\{ \sum_{i=1}^m \mathbf{u}_i\mathbf{G}_i+\mathbf{b}~|~ \mathbf{u}_i\in A_{\epsilon}^{(k_i)}(U_i), i\in [1:m] \},
\end{align}
where $\mathbf{b}\in \ZZ_{p^r}^n$. Note, in this case, we consider $\mathcal{U}$ as a Cartesian product of the typical sets of $U_i$, i.e., $\mathcal{U}=\bigotimes_{i=1}^m A_{\epsilon}^{(k_i)}(U_i).$ 

\begin{definition}\label{def: QGC}
An  $(n,m, k_1,k_2, \dots, k_m)$ QGC over $\ZZ_{p^r}$ is defined as in (\ref{eq: codebook generation}) and is characterized by a translation $\mathbf{b}\in \ZZ^n_{p^r}$, random variables $U_i$ over $\ZZ_{p^r}$  and $k_i \times n$ matrices $\mathbf{G}_i$, where $i\in [1:m]$. 
\end{definition}
\begin{remark}
 Any group code and any transversal group code over $\ZZ_{p^r}$ is a QGC.
\end{remark}

Fix $n,m, k_1,k_2, \dots, k_m$ and random variables $U_i, i\in[1:m]$. We create an ensemble by taking the collection of all $(n, k_1,k_2, \dots, k_m)$ quasi group codes with random variables $U_i$,  for all matrices $\mathbf{G}_i$ and translations $\mathbf{b}$. A random codebook $\mathcal{C}$, from this ensemble, is chosen by selecting the elements of $\mathbf{G}_i, i\in[1:m]$ and $\mathbf{b}$ randomly and uniformly from $\ZZ_{p^r}$. For large enough $n$, with probability close to one the rate of this code is
\begin{align}\label{eq: QGC_rate}
R=\frac{1}{n}\log_2|\mathcal{C}| = \sum_{i=1}^m \frac{k_i}{n}H(U_i).
\end{align}

\begin{remark}
Let $\mathcal{C}$ be a randomly selected QGC as in the above. In contrast to linear codes, codewords of $\mathcal{C}$ are not pairwise independent. 
\end{remark}
We use a different notation to simplify (\ref{eq: QGC_rate}). Let $k=\sum_{i=1}^m k_i$. Denote $q_i=\frac{k_i}{k}$. Since $q_i\geq 0$ and $\sum_i q_i=1$, we can define a random variable $Q$ with $P(Q=i)=q_i$. Define a random variable $U$ with  the conditional distribution $P(U=a|Q=i)=P(U_i=a)$ for all $a\in \ZZ_{p^r}, i\in [1:m]$. Therefore, (\ref{eq: QGC_rate}) is simplified to 
\begin{equation}\label{eq: QGC rate simplified}
R=  \frac{k}{n}H(U|Q).
\end{equation}
 
\begin{remark}
The map induced by the matrices $\mathbf{G}_i$ and the translation $\bf b$ is injective, with high probability, if 
$$ \frac{k}{n} H(U|Q,[U]_s) \leq (r-s)\log_2p,$$ where  $0\leq s\leq r-1.$ Therefore, it is possible to have an injective map for a QGC when $k>n$.
\end{remark} 
 

\subsection{Unionized Quasi Group Codes }\label{subsec: UQGC}

Note that a randomly generated QGC has uniform distribution over the group $\ZZ_{p^r}$. However, in many communication setups we require application of codes with non-uniform distributions. In the case of group codes, this problem is resolved by constructing a group code first, then the union of different shifts of this group code is considered as the codebook. In other words, a large codebook is binned, where the bins themselves are required to possess a group structure \cite{Aria_group_codes}. This new codebook is called a unionized group code. Dual to this codebook construction method, we design a new ensemble of codes. The new codes are called Unionized Quasi Group Codes (UQGC).

A UQGC consists of an inner code and an outer code.  Suppose $\mathcal{C}_{in}$ is a $(n, m, k_1, \dots, k_m)$ QGC with translation $\mathbf{b}$, random variables $U_i$ and matrices $\mathbf{G}_i, i\in [1:m]$. We use $\mathcal{C}_{in}$ as the inner code. 
Given a positive integer $l$, consider a map $t : [1:l]\rightarrow \ZZ_{p^r}^n$. Define the outer code as 
\begin{align}\label{eq: UQGC}
\mathcal{C}_{out}=\bigcup_{j\in [1: l]} (\mathcal{C}_{in}+t(j))
\end{align}

\begin{definition}
Let $\mathcal{C}_{in}$ be an $(n,m, k_1, \dots, k_m)$ QGC. An $(n, m, l, k_1,k_2, \dots, k_m)$ UQGC over $\ZZ_{p^r}$ is defined as in (\ref{eq: UQGC}) and is characterized by $\mathcal{C}_{in}$ as the inner code and a  mapping $t : [1: l] \rightarrow \ZZ_{p^r}^n$. 
\end{definition}

\subsection{Properties of Quasi Group Codes}
It is known that if $\mathcal{C}$ is a random unstructured codebook, then $|\mathcal{C}+\mathcal{C}|\approx|\mathcal{C}|^2$ with high probability. Group codes on the other hand are closed under the addition, which means $|\mathcal{C}+\mathcal{C}|=|\mathcal{C}|$. Comparing to unstructured codes, when the structure of the group codes matches with that of a multi-terminal channel/source coding problem, higher/lower transmission rates are obtained. However, in certain problems, the structure of the group codes is too restrictive. More precisely, when the underlying group is $\ZZ_{p^r}$ for $r\geq 2$, there are several nontrivial subgroups. These subgroups cause a penalty on the rate of a group code. This results in lower transmission rates in channel coding and higher transmission rates in source coding. 

Quasi group codes balance the trade-off between the structure of the group codes and that of the unstructured codes. More precisely, when $\mathcal{C}$ is a QGC, then $|\mathcal{C}+\mathcal{C}|$ is a number between $|\mathcal{C}|$ and $|\mathcal{C}|^2$. This results in a more flexible algebraic structure to match better with the structure of the channel or source. This trade-off is shown more precisely in the following lemma. 

\begin{lem}\label{lem:sum of two quasi group code}
Let $\mathcal{C}$ and $\mathcal{C}'$ be two  $(n,m, k_1, \dots, k_m)$ QGC with random variables $U_i$ and $U'_i, i\in [1:m]$, respectively. Suppose  $\mathcal{C}$ and $\mathcal{C}'$ have identical matrices and translation with elements chosen randomly and uniformly over $\ZZ_{p^r}$. Then for large enough $n$, with probability one, the followings hold:
\begin{enumerate}
\item$\mathcal{C}+a \mathcal{C}'$ is a  $(n,m, k_1, \dots, k_m)$ QGC with random variables $U_i+aU'_i$,
\item $ \max\{| \mathcal{C}|, |a \mathcal{C}'| \} \leq |\mathcal{C}+a \mathcal{C}'| \leq \min\{p^{rn}, |\mathcal{C}| \cdot | a \mathcal{C}'|\} $,
\end{enumerate}
where $a\in \ZZ_{p^r}$ is arbitrary.
\end{lem}

\begin{proof}
The first statement follows from Definition \ref{def: QGC}. For the second statement, strict inequalities follow from standard \textit{arithmetic combinatorics} arguments. Equality in the left-hand side holds, if $U_i$ and $U'_i$ are uniform over $\ZZ_{p^r}$. As for the right-hand side equality, let $U_i$ be uniform over $\{0,2\}$, and $U'_i$ be uniform over $\{0,1\}$. Take $a =1$. Then $U_i+aU'_i$ is uniform over $\ZZ_4$, and  $H(U_i+aU'_i)=2$. Note $H(U_i)=H(U'_i)=1$. Hence, using (\ref{eq: QGC_rate}) the right-hand-side equality holds.  
\end{proof}

In what follows, we derive a packing and a covering bound for a QGC with matrices and translation chosen randomly and uniformly. Fix a PMF $p(x,y)$, and suppose an $\epsilon$-typical sequence $\mathbf{y}$ is given with respect to the marginal distribution $p(y)$.  Consider the set of all codewords that are jointly typical with $\mathbf{y}$ with respect to $p(x,y)$. In the packing lemma, we characterize the conditions in which the probability of this set is small. This implies the existence of a  ``good-channel" code which is also a QGC. In the covering lemma, we derive the conditions for which, with high probability, there exists at least one such codeword. In this case a ``good-source" code exists which is also a QGC. These conditions are provided in the next two lemmas. 

Let $\mathcal{C}$ be a $(n,m, k_1 ,k_2 , \dots, k_m )$ QGC with random variables $U_i$. Suppose the generator matrices and the translation vector of $\mathcal{C}$ are chosen randomly and uniformly over $\ZZ_{p^r}$. Index codewords of $\mathcal{C}$ by $\theta\in [1:|\mathcal{C}|]$. By $\mathbf{c}(\theta)$ denote the $\theta$th codeword of $\mathcal{C}$. Define random variables $Q$ and $U$ as in (\ref{eq: QGC rate simplified}), i.e., $P(Q=i)=\frac{k_i}{\sum_i k_i}$ and $P(U=a|Q=i)=P(U_i=a)$, for all $a\in \ZZ_{p^r}, i\in [1:m]$. Let $\mathcal{C}_{out}$ be a $(n,l,  m, k_1, \dots, k_m)$ UQGC with $l=2^{nR_{bin}}$, $\mathcal{C}$ as an inner code and a map $t:[1:2^{nR_{bin}}]\rightarrow \ZZ_{p^r}^n$ which is selected randomly and uniformly. We use these notations in the following lemmas.

\begin{lem}[Packing]\label{lem: packing}
 Let $(X,Y)\sim p(x,y)$, where $X$ is uniform over $\ZZ_{p^r}$. Fix $\theta \in [1:|\mathcal{C}|]$. Let $\tilde{\mathbf{Y}}^n$ be a random sequence distributed according to $\prod_{i=1}^n p(\tilde{y}_i|c_{i}(\theta))$. Suppose, conditioned on $\mathbf{c}(\theta)$, $\tilde{\mathbf{Y}}^n$ is independent of any other codewords in $\mathcal{C}$. Then, as $n\rightarrow \infty$,  $P\{\exists \mathbf{x}\in \mathcal{C}: (\mathbf{x}, \tilde{\mathbf{Y}}^n)\in A_{\epsilon}^{(n)} (X,Y), \mathbf{x}\neq \mathbf{c}(\theta)\}$ is arbitrary close to zero, if
 \begin{align}\label{eq: packing bound}
R < \min_{0 \leq s\leq r-1} \frac{H(U|Q)}{H(U|Q,[U]_s)}\big( \log_2p^{r-s}-H(X|Y[X]_s) \big).
\end{align}
\end{lem}
\begin{proof}
See Appendix \ref{sec: proof of the packing lemma}.
\end{proof}

\begin{lem}[Covering]\label{lem: covering}
Let $(X,\hat{X})\sim p(x,\hat{x})$,  where $\hat{X}$ is uniform over $\ZZ_{p^r}$.  Let $\mathbf{X}^n$ be a random sequence distributed according to $\prod_{i=1}^n p(x_i)$. Then, as $n \rightarrow \infty$, $ P\{ \exists \hat{\mathbf{x}} \in \mathcal{C}_{out}: (\mathbf{X}^n,\mathbf{\hat{x}})\in A_{\epsilon}^{(n)} (X,\hat{X})\} $
is arbitrary close to one, if 
\begin{align}\label{eq: covering bound}
R_{bin} + \frac{H([U]_s|Q)}{H(U|Q)} R > \log_2 p^s-H([\hat{X}]_s|X)
\end{align}
holds for $1\leq s \leq r$.	
\end{lem}
\begin{proof}
See Appendix \ref{sec: proof of the covering lemma}.
\end{proof}

\begin{remark}\label{rem: covering-paking} 
Using QGC's the symmetric channel capacity and symmetric rate-distortion function are achievable. To see this, set $R_{bin}=0, m=1$ and $U_1$ uniform over $\{0,1\}$.  
\end{remark}

One application of these lemmas is in PtP source coding and channel coding using quasi group codes.
\begin{lem}
Using UQGC's, the PtP channel capacity and rate-distortion function is achievable for channels and sources with alphabet sizes equal to a prime power.
\end{lem}
 \begin{proof}[Outline of the proof]
Consider a memoryless channel with input alphabet $\mathcal{X}$ and conditional distribution $p(y|x)$. Suppose $|\mathcal{X}|$ equals a prime power $p^r$. Fix a PMF $p(x)$ on $\mathcal{X}$, and set $l=2^{nR_{bin}}$. Let $\mathcal{C}_{out}$ be a $(n,l,m=1,k_1)$ UQGC with a mapping $t:[1:l]\rightarrow \ZZ_{p^r}^n$. Let $\mathcal{C}$ be a $(n,m=1,k_1) $ QGC with random variable $U_1$ which is uniform over $\{0,1\}$. We use $\mathcal{C}$ as the inner code for $\mathcal{C}_{out}$. 

 Upon receiving a message $j \in [1:l]$, the encoder finds $c \in \mathcal{C}$ such that $\mathbf{c}+t(j)$ is typical with respect to the PMF $p(x)$. If such $c$ is found, the encoder sends $\mathbf{c}+t(j)$ to the channel; otherwise an encoding error will be declared. Upon receiving $\mathbf{y}$ from the channel, the decoder finds $\tilde{\mathbf{c}} \in \mathcal{C}$ and  $\tilde{j}$ such that $\tilde{\mathbf{c}}+t(\tilde{j})$ is jointly typical with $\mathbf{y}$ with respect to $p(x)p(y|x)$. A decoding error occurs if no unique $\tilde{j}$ is found.  Note the effective transmission rate is $R_{bin}$.

 Let $R_{in}$ be the rate of $\mathcal{C}$. Then, using Lemma \ref{lem: covering}, the probability of the error at the encoder approaches zero, if $R_{in} \geq \log p^r-H(X)$. Using Lemma \ref{lem: packing}, we can show that the average probability of error at the decoder approaches zero, if $R_{in}+R_{bin} \leq \log p^r -H(X|Y)$. As a result the rate $R_{bin}\leq I(X;Y)$ is achievable. 
 
For the source coding problem, let $X\sim p(x)$ be a discrete memorelyss source. Suppose $d: \mathcal{X}\times \mathcal{\hat{X}}\rightarrow [0,+\infty)$ is a distortion function, where the reconstruction alphabet is $\mathcal{\hat{X}}=\ZZ_{p^r}$. Given a distortion level $D$, consider a random variable $\hat{X}$ such that $\EE\{d(X, \hat{X})\}\leq D$. 

Let $\mathbf{x}$ be a typical sequence from the source.  The encoder finds $\mathbf{c}\in \mathcal{C}$ and $j\in [1:l]$ such that $\mathbf{c}+t(j)$ is jointly typical with $\mathbf{x}$ with respect to $p(x)p(\hat{x}|x)$. Then it sends $j$.  If no such $\mathbf{c}$ and $j$ are found, an encoding error will be declared. Given $j$, the decoder finds $\tilde{\mathbf{c}}$ such that  $\tilde{\mathbf{c}}+t( j)$ is typical with respect to $p(\hat{x})$. An error occurs, if no unique codeword $\tilde{\mathbf{c}}$ is found. It can be shown that the encoding error approaches zero, if $R_{bin}+R_{in} \geq \log p^r-H(\hat{X}|X)$. Also the decoding error approaches zero, if $R_{in} \leq \log p^r -H(\hat{X})$. As a result the rate $R_{bin} \geq I(X;\hat{X})$ and distortion $D$ is achievable. 
 
\end{proof}    

Lemma \ref{lem:sum of two quasi group code},  \ref{lem: packing} and Lemma \ref{lem: covering} provide a tool to derive inner bounds for achievable rates using quasi group codes in multi-terminal channel coding and source coding problem. In the next two sections, we study applications of quasi group codes in distributed source coding as well as computation over MAC.

	\section{Distributed Source Coding} \label{sec: dist}
In this section, we consider a special distributed source coding problem. Suppose $X_1$ and $X_2$ are sources over $\ZZ_{p^r}$ with joint PMF $p(x_1,x_2)$. The $j$th encoder compresses $X_j$ and sends it to a central decoder. The decoder wishes to reconstruct $X_1+X_2$ losslessly. 

We use UQGC's to propose a coding strategy for this problem. We use two UQGC's with identical matrices, one for each encoder. As discussed in Subsection \ref{subsec: UQGC}, each UQGC consists of an inner code and an outer code. We select an outer code that is also a  "good-source" code. Consider the codebook created by the sum of the two inner codes. Since, the decoder wishes to reconstruct only $X_1+X_2$, we select the inner codes such that this codebook is also a "good-channel" code. In the following theorem, we characterize an achievable rate region for the above problem using UQGC's.

\begin{theorem}\label{them: distributed source coding}
Suppose $X_1$ and $X_2$ are a pair of sources over the group $\ZZ_{p^r}$.  Lossless reconstruction of $X_1+X_2$ is possible, if the following holds
\begin{align}\label{eq: achievable bounds dist}
R_i \geq \log_2p^r- \frac{H(W_i|Q)}{H(W|[W]_sQ)} (\log_2p^{(r-s)}-H(X|[X]_s)),
\end{align}
where, $i=1,2$, $0 \leq s \leq r-1$, $W=W_1+W_2$ with probability one and the Markov chain $W_1-Q-W_2$ holds.  
 \end{theorem}
\begin{remark}
One can bound the cardinality of $Q$ by $|\mathcal{Q}|\leq r$. This implies that a UQGC with at most $r$ layers is enough to achieve the above bounds. 
\end{remark}

\begin{proof}[Outline of the proof]
Fix positive integers $n, m ,  k_1,  \dots,  k_m$. Let $\mathcal{C}_{1,in}$ and $\mathcal{C}_{2,in}$ be two $(n,m, k_1, \dots, k_m)$ QGC's (as in Definition \ref{def: QGC}) with identical matrices and translation, but with independent random variables. Suppose that the elements of the matrices and translation corresponding to $\mathcal{C}_{1,in}$ and $\mathcal{C}_{2,in}$ are selected randomly and uniformly from $\ZZ_{p^r}$.  
Let $t_1:[1:2^{nR_1}] \rightarrow \ZZ_{p^r}^n$ and $t_2:[1:2^{nR_2}] \rightarrow \ZZ_{p^r}^n$ be two maps selected randomly uniformly and independently of other random variables. 
 
\paragraph*{\textbf{Codebook Generation}} 

%
We use two UQGC's, one for each encoder. The codebook for the first encoder is a $(n,m, l_1, k_1, \dots, k_m)$ UQGC with $l_1=2^{nR_1}$, the inner code $\mathcal{C}_{1,in}$ and the mapping $t_1$.  For the second encoder use a $(n,m, l_2, k_1, \dots, k_m)$ UQGC with $l_2=2^{nR_2}$, the inner code $\mathcal{C}_{2,in}$ and the mapping $t_2$. For the decoder, we use $\mathcal{C}_{1, in}+\mathcal{C}_{2, in}$ as a codebook.

\paragraph*{\textbf{Encoding}}

Given a typical sequence $\mathbf{x}_1\in A_{\epsilon}^n(X_1)$, encoder 1 first finds $i\in [1:2^{nR_1}]$ and $\mathbf{c}_1\in \mathcal{C}_{1,in}$ such that $\mathbf{x}_1=\mathbf{c}_1+t_1(i)$; then it sends $i$. If no such $i$ is found, an error event $E_1$ will be declared. 

Similarly, upon receiving $\mathbf{x}_2 \in A_{\epsilon}^n(X_2)$, the second encoder finds $j\in [1:2^{nR_2}]$ and $\textbf{c}_2\in \mathcal{C}_{2,in}$ such that $\textbf{x}_2=\textbf{c}_2+t_2(j)$ and sends $j$. If no such $j$ is found, an error event $E_2$ will be declared. If more than one indices were found at each encoder, select one randomly. 

\paragraph*{\textbf{Decoding}}
The decoder wishes to reconstruct $\mathbf{x}_1+\mathbf{x}_2$. Assume there is no encoding error.  Upon receiving $i$ and $j$, the decoder first calculates $t_1(i)$ and $t_2(j)$. Then it finds $\tilde{\mathbf{c}}\in \mathcal{C}_{1, in}+\mathcal{C}_{2, in}$ such that $\tilde{\mathbf{c}}+t_1(i)+t_2(j) \in A_{\epsilon}^{(n)}(X_1+X_2)$. If such $\tilde{\mathbf{c}}$ is found, then $\tilde{\mathbf{c}}+t_1(i)+t_2(j)$ is declared as a reconstruction of $\mathbf{x}_1+\mathbf{x}_2$. An error event $E_d$ occurs, if no unique $\tilde{\mathbf{c}}$ was found. 

Using standard arguments for large enough $n$, we can ignored the event in which $\mathbf{x}_1$ and $\mathbf{x}_2$ are not typical. Note that the event $E_i$ is the same as the interested event in Lemma \ref{lem: covering}, where $\hat{X}=X=X_i$ with probability one, $\mathcal{C}_{out}=\mathcal{C}_i$, $R_{bin}=R_i$, $\mathcal{C}=\mathcal{C}_{in, i}$ and $R=R_{in, i}, i=1,2$. Therefore, applying Lemma \ref{lem: covering},  $\mathcal{C}_1$ and $\mathcal{C}_2$ need to satisfy (\ref{eq: covering bound}).  Using Lemma \ref{lem: packing}, we can show that $P(E_d)\rightarrow 0$ as $n\rightarrow \infty$, if the bounds in (\ref{eq: packing bound}) are satisfied with $Y=\emptyset, X=X_1+X_2$ and $\mathcal{C}=\mathcal{C}_{1,in}+\mathcal{C}_{2,in}$. Using the above argument, and noting that the effective transmission rate of the $i$th encoder is  $R_i$, we can derive the bounds in (\ref{eq: achievable bounds dist}). The cardinality bound on $\mathcal{Q}$ and the complete proof are provided in Appendix \ref{sec: proof dist}. 
\end{proof}


Since every linear code,  group code and transversal group code is a QGC,  their achievable rates are included in the rate region characterized in (\ref{eq: achievable bounds dist}). We show, through the following example, that UQGC's improves upon the previously known schemes.

%

\begin{example}
Consider a distributed source coding problem in which $X_1$ and $X_2$ are sources over $\ZZ_4$ and lossless reconstruction of $X_1+X_2$ is required at the decoder.  Assume $X_1$ is uniform over $\ZZ_4$. $X_2$ is related to $X_1$ via $X_2=N-X_1$, where $N$ is independent of $X_1$. The distribution of $N$ is given in Table \ref{tab: N}.

\begin{table}[h]
\caption {Distribution of $N$}\label{tab: N}
\begin{center}
\begin{tabular}{|c|c|c|c|c|}
\hline
N & 0 & 1 & 2 & 3\\
\hline
$P_N$ & $0.1\delta_N$ & $0.9\delta_N$ & $0.1(1-\delta_N)$ & $0.9(1-\delta_N)$\\
\hline
\end{tabular}
\end{center}
\end{table}

Using standard unstructured codes, the rates $R_1+R_2\geq H(X_1,X_2)$ are achievable. As is shown in \cite{Aria_group_codes}, group codes in this example outperform linear codes. The largest achievable region using group codes is $R_j \geq \max \{H(Z), 2 H(Z|[Z]_1)\}, \quad j=1,2, $ where $Z=X_1+X_2$. It is shown in \cite{transversal} that using transversal group codes the rates $$R_j\geq \max \{H(Z), 1/2 H(Z)+ H(Z|[Z]_1)\}$$ are achievable. An achievable rate region using UQGC's can be obtained from Theorem \ref{them: distributed source coding}. Let $Q$ be a trivial random variable and set $P(W_1=0)=P(W_2=0)=0.95$ and  $P(W_1=1)=P(W_2=1)=0.05$. As a result one can verify that the following is achievable: 
$$R_j \geq 2- \min\{ 0.6(2-H(Z)), 5.7 (2-2H(Z|[Z]_1)\}.$$

Let $\delta_N=0.6$. In this case, using unstructured codes the rate $R_i\approx 1.72$ is achievable, using group codes $R_i\approx1.94$ is achievable, using transversal group codes  $R_i\approx1.69$ is achievable. Whereas, $R_i\approx1.67$ is achievable using UQGC's.
\end{example}

\section{Computation Over MAC}\label{sec: comp_over_mac}
Through a variation from the standard computation over MAC problems, in this section, we explore distributed computation of the inputs of a MAC. Figure \ref{fig: comp_over_mac} depicts an example of this problem. Suppose the channel's inputs, $X_1$ and $X_2$, take values from $\ZZ_{p^r}$. Two distributed encoders map their messages to $X^n_1$ and $X^n_2$. Upon receiving the channels output the decoder wishes to decode $X^n_1+ X^n_2$ with no loss. The definition of a code for computation over MAC and an achievable rate are given in Definition \ref{def: code for comp over MAC} and \ref{def: comp over MAC achievable rate}, respectively.  Applications of this problem are in various multi-user communication setups such as interference and broadcast channels. 

\begin{figure}[h]
\centering
\includegraphics[scale=0.8]{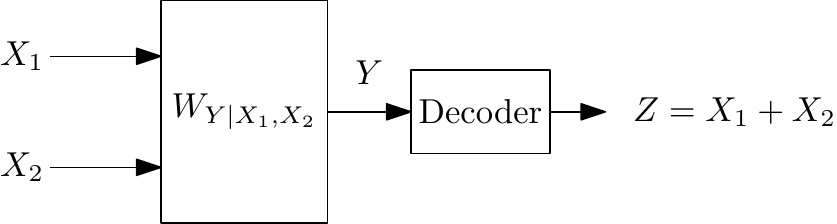} 
\caption{Computation over a two-users MAC.}
\label{fig: comp_over_mac}
\end{figure}


For the above setup, we use quasi group codes to derive an achievable rate region.

\begin{theorem}\label{thm: com_over_mac}
The following is achievable for computation over any MAC with input-alphabets $\ZZ_{p^r}$
\begin{align}\label{eq: comp_over MAC ahiev}
R_i \leq  \frac{H(W_i|Q)}{H(W|[W]_s, Q)} I(X_1+X_2;Y|[X_1+X_2]_s),
\end{align}
where $0\leq s \leq r-1$, $i=1,2$, $X_1$ and $X_2$ are independent and uniform over $\ZZ_{p^r}$, $W=W_1+W_2$, and $W_1-Q-W_2$ holds. Moreover, having $|\mathcal{Q}|\leq r$ is sufficient to achieve the above bounds.
\end{theorem}

\begin{proof}[Outline of the proof]

\textbf{Codebook Generation:} Let $\mathcal{C}_1$ and $\mathcal{C}_2$ be two $(n,m, k_1, k_2, \dots, k_m)$ QGC with identical matrices and independent random variables. Suppose the translations are chosen randomly, independently and uniformly over $\ZZ_{p^r}^n$. Let the elements of the matrices be chosen randomly and uniformly over $\ZZ_{p^r}$. Index all the codewords in each codebooks.

\textbf{Encoding:}  Upon receiving a message index $\theta_j$, encoder $j$ sends the corresponding codeword in $\mathcal{C}_j$, where $j=1,2$. Suppose the output of encoder $j$ is $\mathbf{x}_j, j=1,2$.

\textbf{Decoding:}
Upon receiving $\mathbf{y}$ from the channel, the decoder wishes to decode $\mathbf{x}=\mathbf{x}_1+\mathbf{x}_2$. It finds $\tilde{\mathbf{x}} \in \mathcal{C}_1+\mathcal{C}_2$ such that $\tilde{\mathbf{x}}$ and $\mathbf{y}$ are jointly typical with respect to the distribution $P_{X_1+X_2, Y}$, where $X_1$ and $X_2$ are independent and uniform over $\ZZ_{p^r}$. An error event $E$ is declared, if no unique $\tilde{\mathbf{x}}$ was found. 

Note that by Lemma \ref{lem:sum of two quasi group code}, $\mathcal{C}_1+\mathcal{C}_2$ is a QGC. Using Lemma \ref{lem: packing}, we can show that $P(E)\rightarrow 0$ as $n\rightarrow \infty$, if the bounds in (\ref{eq: packing bound}) hold, where $X=X_1+X_2$ and $\mathcal{C}=\mathcal{C}_1+\mathcal{C}_2$.   Multiply both sides of this bound by $\frac{\log|\mathcal{C}_i|}{\log|\mathcal{C}_1+\mathcal{C}_2|}$. This in turn implies (\ref{eq: comp_over MAC ahiev}), where $P(Q=q) $ is a rational number and  $R_i=\frac{1}{n} \log|\mathcal{C}_i|$.  The complete proof is provided in Appendix \ref{sec: proof of achiv_mac}. 
\end{proof}

We show, through the following example, that using QGC one can improve upon the previously known schemes.

%

\begin{example}\label{ex: comp_z4}
Consider the MAC described by $Y=X_1\oplus X_2 \oplus N,$ where $X_1$ and $X_2$ are the channel inputs with alphabet $\ZZ_4$. $N$ is independent of $X_1$ and $X_2$ with the distribution given in Table \ref{tab: N}, where $0\leq \delta_N \leq 1$.

Using standard unstructured codes the rates satisfying $R_1+R_2\leq I(X_1X_2;Y)$ are achievable. It is shown in \cite{Aria_group_codes} that the largest achievable region using group codes is  $R_i \leq  \min \{  I(Z;Y), 2  I(Z;Y|[Z]_1)\}$, where $Z=X_1+X_2$ and $X_1$ and $X_2$ are uniform over $\ZZ_4$.  It is shown in \cite{transversal} that transversal group codes achieve $$ R_i \leq \min \{ I(Z;Y), 0.5I(Z;Y)+ I(Z; Y| [Z]_1)\}.$$ Using Theorem \ref{thm: com_over_mac}, QGC's achieve
$$R_i \leq \min \{ 0.6 I(Z;Y), 5.7 I(Z; Y| [Z]_1)\}.$$
This can be shown by setting $Q$ to be a trivial random variable, $P(W_1=0)=P(W_2=0)=0.95$ and  $P(W_1=1)=P(W_2=1)=0.05$.

Let $\delta_N=0.6$. Then $R_i \approx 0.28$ is achievable using unstructured codes, $R_i\approx 0.06$ is achievable using group codes and $R_i\approx 0.31$ is achievable using transversal group codes. Whereas, $R_i\approx 0.33$ is achievable using QGC's. 
\end{example}

\section{Conclusion}\label{sec: conclusion}
The problem of computing modulo prime-power was considered. A new layered ensemble of structured codes called QGC was introduced. We investigated the performance limits of these codes in distributed source coding and computation over MAC.  Achievability results using these codes were provided for both settings. We showed that the application of QGC's for these problems results in improvements in terms of transmission rates.

\appendices
\section{Proof of Lemma \ref{lem: packing}} \label{sec: proof of the packing lemma}

Let $\mathcal{C}$ be the random $(n,m,k_1, \dots, k_m)$ QGC as in Lemma \ref{lem: packing}. For shorthand, denote  $\mathcal{U}=\bigotimes_{i=1}^m A_\epsilon^{(k_i)}(U_i)$.  For any $\mathbf{u}_1^m\in \mathcal{U}$, denote $\Phi(\mathbf{u}_1^m)=\sum_{i=1}^m \mathbf{u}_i\mathbf{G}_i$, where $\mathbf{G}_i$'s are random matrices corresponding to $\mathcal{C}$. Fix $\mathbf{u}_0 \in \mathcal{U}$. Without loss of generality assume $\mathbf{c}(\theta)=\Phi(\mathbf{u}_0)+B$, where $B$ is the translation associated with $\mathcal{C}$. Define $\mathcal{E}(\mathbf{u}_1^m):=\{(\Phi(\mathbf{u}_1^m),\tilde{\mathbf{Y}})\in A_\epsilon^{(n)}(X,Y)\}$,  and let $\mathcal{E}$ be the interested event as in the lemma. Then $\mathcal{E}$ is the union of $\mathcal{E}(\mathbf{u}_1^m)$ for all $\mathbf{u}_1^m\in \mathcal{U}\backslash \{\mathbf{u}_0\}$. By the union bound, the probability of $\mathcal{E}$ is bounded as 
\begin{align}\label{eq: union bound on P(E)}
P(\mathcal{E})\leq \sum_{\substack{\mathbf{u}_1^m\in \mathcal{U}\\ \mathbf{u}\neq \mathbf{u}_0}} P(\mathcal{E}(\mathbf{u}_1^m))
\end{align}

The probability of  $\mathcal{E}(\mathbf{u}_1^m)$, can be calculated as,
\begin{align}\label{eq: inner_probability}
P(\mathcal{E}(\mathbf{u}_1^m)) &\approx \sum_{\mathbf{x}_0\in \ZZ_{p^r}^n}  \sum_{\mathbf{y}\in A_\epsilon^{(n)}(Y|\mathbf{x}_0)}P(\Phi(\mathbf{u}_0)+B=\mathbf{x}_0,\tilde{\mathbf{Y}}=\mathbf{y}, \mathcal{E}(\mathbf{u}_1^m))\\\label{eq: inner_probability_eq2}
&=\sum_{(\mathbf{x}_0, \mathbf{y})\in A_\epsilon^{(n)}(X,Y)}  \sum_{\mathbf{x}\in A_\epsilon^{(n)}(X|\mathbf{y})} P(\Phi(\mathbf{u}_0)+B=\mathbf{x}_0,\tilde{\mathbf{Y}}=\mathbf{y}, \Phi(\mathbf{u}_1^m)+B=\mathbf{x} )
\end{align} 
By assumption, conditioned on $\Phi(\mathbf{u}_0)+B$, the random variable $\tilde{\mathbf{Y}}$ is independent of $\Phi(\mathbf{u}_1^m)+B$.  Therefore, the most inner term in (\ref{eq: inner_probability_eq2}) is simplified to 
\begin{equation}
P(\Phi(\mathbf{u}_0)+B=\mathbf{x}_0, \Phi(\mathbf{u}_1^m)+B=\mathbf{x}) p(\mathbf{y}|\mathbf{x}_0).
\end{equation}
Since $B$ is uniform  over $\ZZ_{p^r}^n$, and is independent of other random variables, 
\begin{equation}
P(\Phi(\mathbf{u}_0)+B=\mathbf{x}_0, \Phi(\mathbf{u}_1^m)+B=\mathbf{x})=p^{-nr}P(\Phi(\mathbf{u}_1^m-\mathbf{u}_0)=\mathbf{x}-\mathbf{x}_0).
\end{equation}


We need the following lemma to proceed. Recall that for any $0\leq s \leq r$, $H^n_s=p^s \ZZ^n_{p^r}$.
\begin{lem} \label{lem: P(phi)}
Suppose that $\mathbf{G}_i, i \in [1:m]$ are matrices with elements generated randomly and uniformly from $\ZZ_{p^r}$. If $\mathbf{u}^m_1 \in H^k_s\backslash H^k_{s+1}$, then 
$$P\{ \sum_{i=1}^m \mathbf{u}_i \mathbf{G}_i=\mathbf{x}\}= p^{-n(r-s)}  \11\{x\in H_s^n\}.$$
\end{lem}
\begin{proof}
This Lemma is an special case of Lemma 3 in \cite{Aria_group_codes}. Therefore, we omit the proof of this Lemma.  
\end{proof}

Using the above lemma, if $\mathbf{u}^m_1-\mathbf{u}_0 \in H^k_s\backslash H^k_{s+1}$, then

 \begin{align}\label{eq: inner_probability_cont}
P(\mathcal{E}(\mathbf{u}_1^m))&=  \sum_{(\mathbf{x}_0, \mathbf{y}) \in A_\epsilon^{(n)}(X,Y)} \sum_{\substack{\mathbf{x}\in A_\epsilon^{(n)}(X|\mathbf{y})\\ \mathbf{x}-\mathbf{x}_0\in H_s^n}} p^{-nr}p(\mathbf{y}|\mathbf{x}_0)p^{-n(r-s)} \\\label{eq: inner_probability_cont_eq2}
&=  \sum_{(\mathbf{x}_0, \mathbf{y}) \in A_\epsilon^{(n)}(X,Y)} p^{-nr}p(\mathbf{y}|\mathbf{x}_0)p^{-n(r-s)} |A_\epsilon^{(n)}(X|\mathbf{y})\cap (\mathbf{x}_0+ H_s^n)|
\end{align} 



To calculate $|A_\epsilon^{(n)}(X|\mathbf{y})\cap (\mathbf{x}_0+ H_s^n)|$, the following lemma is needed.
  
\begin{lem}\label{lem: typical set intersection subgroup}
Consider any typical sequences $(\mathbf{\tilde{x}}, \mathbf{y}) \in A_{\epsilon}^{(n)}(X,Y)$. Let $\mathcal{A}=\{\mathbf{x} ~ | ~ \mathbf{x}\in A_\epsilon^n(X|\mathbf{y}), \mathbf{x}-\mathbf{\tilde{x}} \in H^n_{s}\}.$
Then 
\begin{align*}
(1-\epsilon)2^{nH(X|Y [X]_{s})-O(\epsilon)}\leq |\mathcal{A}|\leq 2^{nH(X|Y [X]_{s})+O(\epsilon)}
\end{align*}
\end{lem}  
 \begin{proof}
Refer to Lemma 4 in \cite{Aria_group_codes}.
\end{proof} 
Using the above lemma,  (\ref{eq: inner_probability_cont_eq2}) is bounded as,
\begin{align*}
 & \sum_{(\mathbf{x}_0, \mathbf{y}) \in A_\epsilon^{(n)}(X,Y)} p^{-nr}p(\mathbf{y}|\mathbf{x}_0)p^{-n(r-s)} 2^{nH(X|Y [X]_{s})+O(\epsilon)} \leq p^{-n(r-s)} 2^{nH(X|Y [X]_{s})+O(\epsilon)}.
\end{align*}
Therefore, if $\mathbf{u}^m_1-\mathbf{u}_0 \in H^k_s\backslash H^k_{s+1}$, then $$P(\mathcal{E}(\mathbf{u}_1^m))\leq  p^{-n(r-s)} 2^{nH(X|Y [X]_{s})+O(\epsilon)}. $$
Note that $\ZZ_{p^r}^k$ can be written as disjoint union of $\bigcup_{s=0}^r (H_s^k \backslash H_{s+1}^k)$. Since $\mathbf{u}_1^m\neq \mathbf{u}_0$, we can remove the case $s=r$. Therefore, there are $r$ different cases for each value of $s$. Using (\ref{eq: union bound on P(E)}), and considering these cases, we obtain 

\begin{align*}
P(\mathcal{E}) & \leq \sum_{s=0}^{r-1}\sum_{\substack{\mathbf{u}_1^m\in \mathcal{U}\\ \mathbf{u}_1^m-\mathbf{u}_0 \in  H^k_s \backslash H^k_{s+1}  }} P(\mathcal{E}(\mathbf{u}_1^m)) \leq \sum_{s=0}^{r-1}\sum_{\substack{\mathbf{u}_1^m\in \mathcal{U}\\ \mathbf{u}_1^m-\mathbf{u}_0 \in  H^k_s \backslash H^k_{s+1}  }}  2^{nH(X|Y[X]_s)}p^{-n(r-s)}\\
&\leq  \sum_{s=0}^{r-1}|\mathcal{U}\cap (\mathbf{u}_0+H_s^k)| 2^{nH(X|Y[X]_s)}p^{-n(r-s)}
\end{align*}  

Since   $\mathcal{U}=\bigotimes_{i=1}^m A_\epsilon^{(k_i)}(U_i)$, for each component $i$ of $\mathcal{U}$, we can apply Lemma \ref{lem: typical set intersection subgroup} to get,
$$|\mathcal{U}\cap (\mathbf{u}_0+H_s^k)|\leq 2^{\sum_i k_i H(U_i|[U_i]_s)}=2^{kH(U|Q[U]_s)}.$$
As a result $P(\mathcal{E})\rightarrow 0$ as $n\rightarrow \infty$, if $\frac{k}{n}H(U|Q[U]_s) \leq \log_2p^{r-s}-H(X|Y[X]_s)$, for all $0\leq s\leq r-1$.  
Multiply each side of this inequality by $\frac{H(U|Q)}{H(U|Q[U]_s)}$, gives the following bound
$$\frac{k}{n}H(U|Q) \leq \frac{H(U|Q)}{H(U|Q[U]_s)}(\log_2p^{r-s}-H(X|Y[X]_s))$$  
By definition $R=\frac{1}{n}\log_2|\mathcal{C}|\leq \frac{k}{n}H(U|Q)$. Therefore, 
$$R \leq \frac{H(U|Q)}{H(U|Q[U]_s)}(\log_2p^{r-s}-H(X|Y[X]_s)),$$ and the proof is completed.

\section{Proof of Lemma \ref{lem: covering}}\label{sec: proof of the covering lemma}
We use the same notation as in the proof of Lemma \ref{lem: packing}. For any typical sequence $\mathbf{x}$ define
\begin{align*}
\delta(\mathbf{x})= \sum_{\mathbf{\hat{x}}\in A_\epsilon^{(n)}(\hat{X}|\mathbf{x})}  \sum_{\mathbf{u}_1^m \in \mathcal{U}}  \sum_{j=1}^l \mathbbm{1}\{\Phi(\mathbf{u}_1^m)+t(j)=\hat{x}\}.
\end{align*}
Note $\delta(\mathbf{x})$	counts the number of codewords that are conditionally typical with $\mathbf{x}$ with respect to $p(\hat{\mathbf{x}, \mathbf{x}})$. So, if $\delta(\mathbf{x})=0$,  then the complement of the interested event in Lemma \ref{lem: covering} occurs. Hence, it suffices to show that $\lim_{n\rightarrow \infty }P(\delta(\mathbf{x})=0)=0$. We have,
\begin{align*}
P\{\delta(\mathbf{x})=0\} \leq  P\Big\{\delta(\mathbf{x})\leq \frac{1}{2} E(\delta(x))\Big\}  \leq P\Big\{|\delta(x)-E(\delta(x)) |\geq \frac{1}{2} E(\delta(x))\Big\}
\end{align*}
Hence, by Chebyshev's inequality, $P\{\delta(\mathbf{x})=0\}  \leq \frac{4 Var(\delta(x))}{E(\delta(x))^2}$.
Note that 
\begin{align} \label{eq: expectation_delta}
E(\delta(x))=\sum_{\mathbf{\hat{x}}\in A_\epsilon^{(n)}(\hat{X}|\mathbf{x})}  \sum_{\mathbf{u}_1^m \in \mathcal{U}}  \sum_{j=1}^l   P\{\Phi(\mathbf{u}_1^m)+t(j)=\hat{\mathbf{x}}\}
\end{align}
Since $t(j)$ is uniform over $\ZZ_{p^r}^n$, and $l=2^{nR_{bin}}$ we get
\begin{align} \label{eq: expectation_delta}
E(\delta(x))\leq  2^{nH(\hat{X}|X)} | \mathcal{U}| 2^{nR_{bin}} p^{-rn}.
\end{align}
Note $|\mathcal{U}|\approx 2^{kH(U|Q)}$. To calculate the variance, we start with
\begin{align*}
E(\delta(x)^2)&= \sum_{\mathbf{\hat{x},\hat{x}'}\in A_\epsilon^{(n)}(\hat{X}|\mathbf{x})} \sum_{\mathbf{u}_1^m, {\mathbf{u}'}_1^m \in \mathcal{U}}  \sum_{j, j'=1}^l   P\{\Phi(\mathbf{u}_1^m)+t(j)=\mathbf{\hat{x}}, \Phi({\mathbf{u}'}_1^m)+t(j')=\mathbf{\hat{x}'}\}.
\end{align*}
If $j\neq j'$, then $t(j)$ is independent of $t(j')$. Therefore $$ P\{\Phi(\mathbf{u}_1^m)+t(j)=\mathbf{\hat{x}}, \Phi({\mathbf{u}'}_1^m)+t(j')=\mathbf{\hat{x}'}\}=p^{-2nr}.$$ As a result, 
\begin{align*}
E(\delta(x)^2)&\leq   E(\delta(x))^2 +  \sum_{\mathbf{\hat{x},\hat{x}'}\in A_\epsilon^{(n)}(\hat{X}|\mathbf{x})} \sum_{\mathbf{u}_1^m, {\mathbf{u}'}_1^m \in \mathcal{U}}  \sum_{j=1}^l   P\{\Phi(\mathbf{u}_1^m)+t(j)=\mathbf{\hat{x}}, \Phi({\mathbf{u}'}_1^m)+t(j)=\mathbf{\hat{x}'}\}.
\end{align*}

Since $t(j)$ is independent of other random variables, the most inner term in the above summations is simplified to $p^{-nr}P\{\Phi(\mathbf{u}_1^m-\mathbf{u'}_1^m)=\mathbf{\hat{x}}-\mathbf{\hat{x}'}\}$. Using Lemma \ref{lem: P(phi)}, if $\mathbf{u}_1^m-\mathbf{u'}_1^m \in H_s^k\backslash H_{s+1}^k$, then
\begin{align*}
P\{\Phi(\mathbf{u}_1^m-\mathbf{u'}_1^m)=\mathbf{\hat{x}}-\mathbf{\hat{x}'}\}=p^{n(r-s)}\mathbbm{1}\{\mathbf{\hat{x}}-\mathbf{\hat{x}'} \in H_s^n\}
\end{align*}
We use the same argument as in the proof of Lemma \ref{lem: packing}. We divide $\ZZ_{p^r}^k$ into disjoint union of $H_s^k\backslash H_{s+1}^k$. Hence, we get
\begin{align*}
E(\delta(x)^2)&\leq  E(\delta(x))^2 + \sum_{s=0}^{r} \sum_{\substack{\mathbf{u}_1^m, {\mathbf{u}'}_1^m \in \mathcal{U}\\ \mathbf{u}_1^m-{\mathbf{u}'}_1^m\in H^k_{s}\backslash H_{s+1}^k}}\sum_{\substack{\mathbf{\hat{x},\hat{x}'}\in A_\epsilon^{(n)}(\hat{X}|\mathbf{x})\\ \mathbf{\hat{x}}-\mathbf{\hat{x}'} \in H_s^n}}  \sum_{j=1}^l    p^{-nr}p^{-n(r-s)}
\end{align*}
 Since the most inner terms in the above summations do not depend on the individual values of $\mathbf{x}, \hat{\mathbf{x}}, \mathbf{u}_1^m, {\mathbf{u}'}_1^m, j$, the corresponding summations can be replaced by the size of the associated sets. Moreover, we upper-bound the summation over $\mathbf{u}_1^m, {\mathbf{u}'}_1^m$ by replacing $H_s^k\backslash H_{s+1}^k$ with $H_s^k$. Using Lemma \ref{lem: typical set intersection subgroup} for $\mathbf{x}, \hat{\mathbf{x}}$, we get
\begin{align*}
E(\delta(x)^2)&\leq  E(\delta(x))^2 + \sum_{s=0}^{r} \sum_{\mathbf{u}_1^m \in \mathcal{U}} \sum_{\substack{{\mathbf{u}'}_1^m \in \mathcal{U}\\ \mathbf{u}_1^m-{\mathbf{u}'}_1^m\in H^k_{s}}}2^{n(H(\hat{X}|X)+H(\hat{X}|X[\hat{X}]_s))} 2^{nR_{bin}} p^{-nr}p^{-n(r-s)}
\end{align*}

%

Note $\mathcal{U}=\bigotimes A_\epsilon^{(k_i)}(U_i)$. For any $\mathbf{u}_1^m$, using Lemma \ref{lem: typical set intersection subgroup}, $|\mathcal{U} \cap (\mathbf{u}_1^m+H_s^k)|\approx 2^{kH(U|Q[U]_s)}.$
 As a result, 
\begin{align*}
Var(\delta(x)^2)& \leq 2^{nR_{bin}} 2^{kH(U|Q)}p^{-nr} \sum_{s=0}^{r} 2^{kH(U|Q[U]_s)} 2^{n(H(\hat{X}|X)+H(\hat{X}|X[\hat{X}]_s))}  p^{-n(r-s)}
\end{align*}

Finally, using the Chebyshev's inequality argued before, we get

\begin{align*}
P\{\delta(\mathbf{x})=0\}&\leq 4 \sum_{s=0}^r 2^{kH(U|Q[U]_s)} ~ 2^{-kH(U|Q)} 2^{-n(H(\hat{X}|X)-H(\hat{X}|X[\hat{X}]_s))}p^{nr}p^{-n(r-s)} \\
&= \sum_{s=0}^r 2^{-KH([U]_s|Q)} 2^{-nH([\hat{X}]_s|X)} 2^{-nR_{bin}} p^{ns}.
\end{align*}
The second equality follows, because $H(V|W)-H(V|[V]_sW)=H([V]_s|W)$ holds for any random variables $V$ and $W$. Therefore, $P\{\delta(\mathbf{x})\}$ approaches zero, as $n\rightarrow \infty$, if 
\begin{align}\label{eq: covering_dist}
R_{bin}+\frac{k}{n} H([U]_{s}|Q) \geq  \log_2 p^s -H([\hat{X}]_s|X), \quad  \mbox{for} ~~0\leq s\leq r.
\end{align}
Note that the above inequality for the case $s=0$ is trivial.  In the above arguments we assumed that $R_{bin}>0$. One can verify that (\ref{eq: covering_dist 2}) still holds when  $R_{bin}=0$.
By the definition of rate and the above inequalities the proof is completed.

\section{ Proof of Theorem \ref{them: distributed source coding}} \label{sec: proof dist}
We need to find conditions for which the probability of the error events $E_1, E_2$ and $E_d$ is small enough. Without loss of generality we can assume that $\mathbf{b=0}$. Let $k=\sum_{i=1}^m k_i$.  Define the random variable $Q$ with PMF $P(Q=i)=\frac{k_i}{k}$. Define random variables $U$ and $V$ over $\ZZ_{p^r}$ with conditional PMF  $P(U =a | Q=i )=P(U_{i}=a)$ and $P(V=a|Q=i)=P(V_{i}=a), i\in [1:m], a\in \ZZ_{p^r}$. We follow the notation used in Appendix \ref{sec: proof of the packing lemma}. Let $\mathcal{U}=\bigotimes_{i=1}^m A_\epsilon^{(k_i)}(U_{i}), \mathcal{V}=\bigotimes_{i=1}^m A_\epsilon^{(k_i)}(V_{i})  $ and  $\Phi(\mathbf{a}_1^m)=\sum_{i=1}^m \mathbf{a}_i\mathbf{G}_i$, where $\mathbf{a}_i \in \ZZ_{p^r}^{k_i}, i\in [1:m]$.  With this notation, $|\mathcal{U}|\approx 2^{kH(U| Q)}, |\mathcal{V}|\approx 2^{kH(V| Q)}$.

\subsection{Analysis of $E_1, E_2$}
For any typical $\mathbf{x}_1$ define $\delta(\mathbf{x}_1)=\sum_{i=1}^{2^{nR_1}}\sum_{\mathbf{u}_1^m \in \mathcal{U}}  \mathbbm{1}\{\mathbf{x}_1=\phi(\mathbf{u}_1^m)+t(i)\}.$ Therefore, $E_1$ occurs if $\delta(x_1)=0$. For more convenience, we alleviate the event $E_1$. Assume $E_1$ occurs, if $\delta(\mathbf{x})< \frac{1}{2}E(\delta(x_1))$. We use Lemma \ref{lem: covering} to show that $P(E_1)$ is small enough. In this lemma set $\hat{X}=X=X_1$ with probability one, $\mathcal{C}_{out}=\mathcal{C}_1$, $R_{bin}=R_1$, $\mathcal{C}=\mathcal{C}_{in, 1}$ and $R=R_{in, 1}$. Therefore, $P(E_1)\rightarrow 0$ as $n\rightarrow \infty$, If 
\begin{align}\label{eq: covering_dist}
R_1+\frac{k}{n} H([U]_{s}|Q) \geq \log_2 p^s. 
\end{align}

Similarly for a typical sequence $\mathbf{x}_2 \in A_\epsilon^{(n)}(X_2)$, define $\eta(\mathbf{x}_2)=\sum_{j=1}^{2^{nR_2}}\sum_{\mathbf{v}_1^m \in \mathcal{V}}  \mathbbm{1}\{\mathbf{x}_2=\phi(\mathbf{v}_1^m)+\tau(j)\}.$
Notice $E_2$ occurs if $\eta(y)=0$. For more convenience, we assume $E_2$ occurs if $\eta(\mathbf{x}_2 )< \frac{1}{2}E(\eta(\mathbf{x}_2 ))$. Using a similar argument for $E_1$, we can show that $P\{E_2\}$ approaches zero, as $n \rightarrow \infty$, if 
\begin{align}\label{eq: covering_dist 2}
R_2+\frac{k}{n} H([V]_{s}|Q) \geq \log_2 p^s.
\end{align}

\subsection{Analysis  of $E_d$}
Assume there is no error in the encoding stage. Also suppose  the indices $i$ and $j$ are sent by the first and the second encoders, respectively. $E_d$ occurs at the decoder, if there exists  $\tilde{\mathbf{c}} \in \mathcal{C}_{1, in}+\mathcal{C}_{2, in}$ such that  $\tilde{\mathbf{z}}=\tilde{\mathbf{c}}+t(i)+\tau(j)$  is typical with respect to $P_{X_1+X_2}$ and $\tilde{\mathbf{z}}\neq \mathbf{x}_1+\mathbf{x}_2$. Since we assumed that there is no error at the encoders, we are looking for the event $E=E_d \cap E_1^c \cap E_2^c$. Suppose all the codebooks are fixed. This means $\phi(\cdot), t(\cdot)$ and $\tau(\cdot)$ are fixed. Conditioned on $\mathbf{x}_1, \mathbf{x}_2, i$ and $j$, the probability of $E$ equals to
\begin{align*}
P(E| \mathbf{x}_1, \mathbf{x}_2, i, j)=\mathbbm{1}\{\exists \tilde{z}\in A_\epsilon^{(n)}(X_1+X_2): \tilde{z}\neq \mathbf{x}_1+\mathbf{x}_2, \tilde{z}\in \mathcal{C}_{1, in}+\mathcal{C}_{2, in}+t(i)+\tau(j)\}
\end{align*}
Suppose the pairs $(i,\mathbf{u}_1^m)$ and $(j,\mathbf{v}_1^m)$ are chosen at the encoders. In what follows, we bound $P(E| \mathbf{x}_1, \mathbf{x}_2, i, j)$. Let  $\mathcal{W}$ be the union of $\tilde{\mathbf{u}}_1^m+\tilde{\mathbf{v}}_1^m$, for all $\tilde{\mathbf{u}}_1^m \in \mathcal{U}$ and $\tilde{\mathbf{v}}_1^m \in \mathcal{V}$. We denote $\mathcal{W}=\mathcal{U}+\mathcal{V}$. Also, define $Z=X_1+X_2$, with probability one. Using the union bound, $P(E| \mathbf{x}_1, \mathbf{x}_2, i, j)$ is bounded as,
\begin{align}\label{eq: bound on P(E|X_1X_2...)}
P(E| \mathbf{x}_1, \mathbf{x}_2, i, j) &\leq   \sum_{\substack{ \tilde{\mathbf{w}}_1^m \in \mathcal{W}\\ \tilde{\mathbf{w}}_1^m \neq \mathbf{u}_1^m+\mathbf{v}_1^m}}\sum_{ \tilde{\mathbf{z}}\in A_{\epsilon}^{(n)}(Z)} \mathbbm{1}\{\phi(\tilde{\mathbf{w}}_1^m)+t(i)+\tau(j)=\tilde{\mathbf{z}}\}
\end{align}
 
Given $\mathbf{x}_1$, if more than one pair $(i,\mathbf{u}_1^m)$ was found at the first encoder, select one randomly and uniformly. Therefore, conditioned on $\mathbf{x}_1$, a pair $(i,\mathbf{u}_1^m)$ is selected with probability $P(i, \mathbf{u}_1^m|\mathbf{x_1})=\frac{1}{\delta(\mathbf{x}_1)}\mathbbm{1}\{\phi(\mathbf{u}_1^m)+t(i)=\mathbf{x}_1\}.$ Similarly a pair $(j,\mathbf{v}_1^m)$ is selected at the second encoder with probability $P(j, \mathbf{v}_1^m|\mathbf{x}_2)=\frac{1}{\eta(\mathbf{x}_2)}\mathbbm{1}\{\phi(\mathbf{v}_1^m)+\tau(j)=\mathbf{x}_2\}.$ Since there is no encoding error, $\delta(\mathbf{x}_1) \geq \frac{1}{2}E(\delta(\mathbf{x}_1))$ and $\eta(\mathbf{x}_2) \geq \frac{1}{2}E(\eta(\mathbf{x}_2))$. As a result, 
\begin{align}\label{eq: bound on P(i, u|x_1)}
P(i, \mathbf{u}_1^m|\mathbf{x_1})\leq \frac{2}{E(\delta(\mathbf{x}_1))}\mathbbm{1}\{\phi(\mathbf{u}_1^m)+t(i)=\mathbf{x}_1\}\\\nonumber
P(j, \mathbf{v}_1^m|\mathbf{x}_2)\leq \frac{2}{E(\eta(\mathbf{x}_2))}\mathbbm{1}\{\phi(\mathbf{v}_1^m)+\tau(j)=\mathbf{x}_2\}.
\end{align}
By definition, probability of $E$ equals to

\begin{align*}
P(E)&=\sum_{(\mathbf{x}_1, \mathbf{x}_2) \in A_{\epsilon}^{(n)}(X_1X_2)} p(\mathbf{x}_1, \mathbf{x}_2) \sum_{i=1}^{2^{nR_1}} \sum_{\mathbf{u}_1^m \in \mathcal{U}} P(i, \mathbf{u}_1^m |\mathbf{x}_1) \sum_{j=1}^{2^{nR_2}} \sum_{\mathbf{v}_1^m \in \mathcal{V}} P(j, \mathbf{v}_1^m |\mathbf{x}_2)P(E|\mathbf{x}_1,\mathbf{x}_2,i, j)
\end{align*}
 Using the bounds given in (\ref{eq: bound on P(E|X_1X_2...)}) and (\ref{eq: bound on P(i, u|x_1)}), we get

\begin{align*}
P(E)&\leq  \sum_{(\mathbf{x}_1, \mathbf{x}_2) \in A_{\epsilon}^{(n)}(X_1X_2)} p(\mathbf{x}_1, \mathbf{x}_2) \sum_{i=1}^{2^{nR_1}} \sum_{\mathbf{u}_1^m \in \mathcal{U}}\frac{2}{E({\delta}(\mathbf{x}_1))}\mathbbm{1}\{\phi(\mathbf{u}_1^m)+t(i)=\mathbf{x}_1\}\\
&\sum_{j=1}^{2^{nR_2}} \sum_{\mathbf{v}_1^m \in \mathcal{V}}\frac{2}{E(\eta(\mathbf{x}_2))}\mathbbm{1}\{\phi(\mathbf{v}_1^m)+\tau(j)=\mathbf{x}_2\}\sum_{\substack{ \tilde{\mathbf{w}}_1^m \in \mathcal{W}\\ \tilde{\mathbf{w}}_1^m \neq \mathbf{u}_1^m+\mathbf{v}_1^m}}\sum_{ \tilde{z}\in A_{\epsilon}^{(n)}(Z)} \mathbbm{1}\{\phi(\tilde{\mathbf{w}}_1^m)+t(i)+\tau(j)=\tilde{\mathbf{z}}\}
\end{align*}

Note $ p(\mathbf{x}_1, \mathbf{x}_2)\approx 2^{-nH(X_1,X_2)}$. Averaging over all possible choices of $\phi(\cdot), t(\cdot)$ and $\tau(\cdot)$ gives
\begin{align*}
\EE\{P(E)\}&\leq \sum_{(\mathbf{x}_1, \mathbf{x}_2) \in A_{\epsilon}^{(n)}(X_1X_2)}  2^{-nH(X_1,X_2)} \sum_{i=1}^{2^{nR_1}} \sum_{\mathbf{u}_1^m \in \mathcal{U}}\frac{2}{E({\delta}(\mathbf{x}_1))} \sum_{j=1}^{2^{nR_2}} \sum_{\mathbf{v}_1^m \in \mathcal{V}}\frac{2}{E(\eta(\mathbf{x}_2))}\\
& \sum_{\substack{ \tilde{\mathbf{w}}_1^m \in \mathcal{W}\\ \tilde{\mathbf{w}}_1^m \neq \mathbf{w}_1^m}}\sum_{ \tilde{z}\in A_{\epsilon}^{(n)}(Z)} P\{\Phi(\mathbf{u}_1^m)+t(i)=\mathbf{x}_1, \Phi(\mathbf{v}_1^m)+\tau(j)=\mathbf{x}_2, \Phi(\tilde{\mathbf{w}}_1^m)+t(i)+\tau(j)=\tilde{\mathbf{z}}\}
\end{align*}
Note $t(i)$ and $\tau(j)$ are independent random variables and uniformly distributed over $\ZZ_{p^r}^n$. This in turn implies that the most inner term in the above summations equals 
\begin{align*}
{p^{-2nr}}P\{\Phi(\tilde{\mathbf{w}}_1^m-\mathbf{u}_1^m-\mathbf{v}_1^m)=\tilde{\mathbf{z}}-\mathbf{x}_1-\mathbf{x}_2\}.
\end{align*}
Using Lemma \ref{lem: P(phi)}, we can determine the above probability. We have,

\begin{align*}
\EE\{P(E)\}&\leq \sum_{(\mathbf{x}_1, \mathbf{x}_2) \in A_{\epsilon}^{(n)}(X_1X_2)}  2^{-nH(X_1,X_2)} \sum_{i=1}^{2^{nR_1}} \sum_{\mathbf{u}_1^m \in \mathcal{U}}\frac{2}{E({\delta}(\mathbf{x}_1))} \sum_{j=1}^{2^{nR_2}} \sum_{\mathbf{v}_1^m \in \mathcal{V}}\frac{2}{E(\eta(\mathbf{x}_2))}\\
&\sum_{s =0}^{r-1} \sum_{\substack{ \tilde{\mathbf{w}}_1^m \in \mathcal{W}\\ \tilde{\mathbf{w}}_1^m - \mathbf{u}_1^m-\mathbf{v}_1^m\in H^k_{s} }}\sum_{ \substack{ \tilde{\mathbf{z}}\in A_{\epsilon}^{(n)}(Z)\\ \tilde{\mathbf{z}}-\mathbf{x}_1-\mathbf{x}_2 \in H^n_{s}}} p^{-2nr}p^{-n(r-s)}
\end{align*}
Since the most inner terms in the above summations depend only on $s$, we can replace the corresponding summations with the size of the associated sets. Hence, using Lemma \ref{lem: typical set intersection subgroup}, we get,
\begin{align*}
\EE\{P(E)\}&\leq \sum_{(\mathbf{x}_1, \mathbf{x}_2) \in A_{\epsilon}^{(n)}(X_1X_2)}  2^{-nH(X_1,X_2)} \sum_{i=1}^{2^{nR_1}} \sum_{\mathbf{u}_1^m \in \mathcal{U}}\frac{1}{E({\delta}(\mathbf{x}_1))} \sum_{j=1}^{2^{nR_2}} \sum_{\mathbf{v}_1^m \in \mathcal{V}}\frac{1}{E(\eta(\mathbf{x}_2))}\\
&\sum_{s=0}^{r-1} 2^{nH(Z|[Z]_{s})}2^{k H(W|Q [W]_s)} p^{-2nr}p^{-n(r-s)}\\  
&\leq  2^{nR_1} |\mathcal{U}|\frac{1}{E({\delta}(\mathbf{x}_1))}2^{nR_2} |\mathcal{V}| \frac{1}{E(\eta(\mathbf{x}_2))} \sum_{s=0}^{r-1} 2^{nH(Z|[Z]_{s})}2^{ kH(W|Q, [W]_s)} p^{-2nr}p^{-n(r-s)}.
\end{align*}
Note $E({\delta}(\mathbf{x}_1))$ and $E(\eta(\mathbf{x}_2))$ can be bounded as in (\ref{eq: expectation_delta}). Therefore, we have 
\begin{align*}
\EE\{P(E)\}\leq \sum_{s =0}^{r-1}  2^{nH(Z|[Z]_{s})} 2^{k H(W|Q, [W]_s)} p^{-n(r-s)}
\end{align*}
Hence, $\EE\{P(E)\}$ tends to zero as $n\rightarrow \infty$, if for any $s \in [0:r-1]$,
\begin{align}\label{eq: packing_dist}
\frac{k}{n} H(W| Q, [W]_{s}) \leq \log_2 p^{(r-s)}-H(Z|[Z]_{s}).
\end{align}
Note  having (\ref{eq: packing_dist}) the bounds in (\ref{eq: covering_dist}) and (\ref{eq: covering_dist 2}) are redundant except the following:
\begin{align}\label{eq: covering_dist_simplified}
R_1+\frac{k}{n} H(U|Q) =  \log_2 p^r\\\label{eq: covering_dist_simplified 2}
R_2+\frac{k}{n} H(V|Q) =  \log_2 p^r
\end{align}
Lastly using (\ref{eq: covering_dist_simplified}), (\ref{eq: covering_dist_simplified 2}) and (\ref{eq: packing_dist}) the followings are achievable:
\begin{align}\label{eq: dist_achievable}
R_1\geq  \log_2 p^r- \frac{H(U|Q)}{H(W| Q, [W]_{s})} (\log_2 p^{(r-s)}-H(Z|[Z]_{s}))\\\label{eq: dist_achievable 2}
R_2\geq  \log_2 p^r- \frac{H(V|Q)}{H(W| Q, [W]_{s})} (\log_2 p^{(r-s)}-H(Z|[Z]_{s})),
\end{align}
where  we take the union over all PMF $p(u,v,q) = p(q)p(u|q)p(v|q)$, such that $p(q)$ is a rational number, $q\in [1:m]$. Since rational numbers are dense in $\RR$, one can consider arbitrary PMF $p(q)$.  Let $\mathcal{Q}$ be the set over which $Q$ takes values, and $P(Q=q)>0, q\in \mathcal{Q}$.  Note the cardinality of $\mathcal{Q}$ determines $m$ which is the number of layers used in the coding strategy. To achieve the above bounds, we show that $r$ layers is enough, i.e., $|\mathcal{Q}|\leq r$. 

\begin{lem}\label{lem: cardinality of Q}
The cardinality of $\mathcal{Q}$ is bounded by $|\mathcal{Q}|\leq r$. 
\end{lem}

\begin{proof}
Note that (\ref{eq: packing_dist}), (\ref{eq: covering_dist_simplified}) and (\ref{eq: covering_dist_simplified 2}) are an alternative characterization of the achievable region. Using these equations, observe that this region is convex in $\RR^2$. As a result, we can characterize the achievable region by its supporting hyperplanes. Let $\bar{R}_i:= \log_2 p^r -R_i, i=1,2$. Using (\ref{eq: dist_achievable}) and (\ref{eq: dist_achievable 2}) for any $0\leq \alpha \leq 1 $ the corresponding supporting hyperplan is charachterized by 
\begin{align} \label{eq: support hyperplan dist}
\big(\alpha \bar{R}_1 +(1-\alpha)\bar{R}_2 \big) H(W| Q, [W]_{s}) - \Big(\alpha H(U|Q) + (1-\alpha)H(V|Q)\Big) \Big(\log_2 p^{(r-s)}-H(Z|[Z]_{s})\Big)\leq 0, 
\end{align}
where $s\in [0,r-1]$. We use the support lemma for the above inequalities to bound $|\mathcal{Q}|$. To this end, we first show that the left-hand side of these inequalities are continuous functions of conditional  PMF's of $U$ and $V$ given $Q$. Let $\mathscr{P}_r$ denote the set of all product PMF's on $\ZZ_{p^r}\times \ZZ_{p^r}$. Note $\mathscr{P}_r$  is a compact set. Fix $q\in \mathcal{Q}$. Denote $ f(p(u|q)p(v|q))=\alpha H(U|Q=q) + (1-\alpha)H(V|Q=q)$ and $g_s(p(u|q)p(v|q))=  H(U+V| Q=q, [U+V]_{s})$, where $s\in [0:r-1]$. We show that  $f(\cdot), g_s(\cdot)$ are real valued continuous functions of $ \mathscr{P}_r$. Since the entropy function is continuous then so is $f$. We can write  $g_s(p(u|q)p(v|q))= H(U+V|Q=q)-H([U+V]_s|Q=q)$. Note that $[\cdot]_s$ is a continuous function from $ \mathscr{P}_r$ to $\mathscr{P}_r$. This implies that  $H([\cdot]_s)$ is also continuous. So $g_s$ is continuous. As a result, the left-hand side of the bounds in (\ref{eq: support hyperplan dist}) are real valued continuous functions of $\mathscr{P}_r$. Therefore, we can apply the support lemma. Since there are $r$ bounds, then $|\mathcal{Q}|\leq r$. 
\end{proof}


\section{Proof of Theorem \ref{thm: com_over_mac}}\label{sec: proof of achiv_mac}
 We follow the same notation as in the proof of Lemma \ref{lem: packing}. For any $\mathbf{a}_i \in \ZZ_{p^r}^{k_i}, i\in [1:m]$, denote  $\Phi(\mathbf{a}_1^m)=\sum_{i=1}^m \mathbf{a}_i\mathbf{G}_i$. For each $i\in [1:m]$, suppose $U_i$ and $V_i$ are random variables corresponding to $\mathcal{C}_1$ and $\mathcal{C}_2$, respectively.  Let $\mathcal{U}=\bigotimes_{i=1}^m A_\epsilon^{(k_i)}(U_{i}), \mathcal{V}=\bigotimes_{i=1}^m A_\epsilon^{(k_i)}(V_{i})  $.
 
  We change the definition of the error event at the decoder. We require the decoder to decode $\mathbf{u}_i+\mathbf{v}_i, i \in [1:m] $. This is a stronger condition, but it is more convenient for error analysis. In what follows, we redefine the decoding operation. Upon receiving $\mathbf{y}$, the decoder finds $\tilde{\mathbf{w}}_i \in A_{\epsilon}^{(k_i)}(U_i+V_i)$ such that $\Phi(\tilde{\mathbf{w}}_1^m)+b_1+b_2$ is jointly typical with $\mathbf{y}$ w.r.t $P_{X_1+X_2,Y}$.  The error event $E$ occurs at the decoder, if  $\tilde{\mathbf{w}}_1^m$ is not unique.  Let $k=\sum_{i=1}^m k_i$.  As in Appendix \ref{sec: proof dist},  defined a random variable $Q$ with PMF $P(Q=i)=\frac{k_i}{k}$. Define random variables $U$ and $V$ over $\ZZ_{p^r}$ with conditional PMF  $P(U =a | Q=i )=P(U_{i}=a)$ and $P(V=a|Q=i)=P(V_{i}=a), i\in [1:m], a\in \ZZ_{p^r}$.   Assume $\Phi=\phi, \mathbf{b}_1$ and $\mathbf{b}_2$ are fixed,  the probability of $E$ is  
 \begin{align*}
P(E|  \phi, \mathbf{b}_1, \mathbf{b}_2)=&\sum_{u_1^m\in \mathcal{U}}\frac{1}{|\mathcal{U}|} \sum_{x_1\in \ZZ_{p^r}^n} \mathbbm{1}\{ x_1 = \phi(u_1^m)+b_1\} \sum_{v_1^m\in \mathcal{V}}\frac{1}{|\mathcal{V}|} \sum_{x_2\in \ZZ_{p^r}^n} \mathbbm{1}\{ x_2= \phi(v_1^m)+b_2\} \sum_{y\in A_\epsilon^n(Y|x_1,x_2)}  P(y|x_1, x_2)\\
& \mathbbm{1}\{ \exists ~ \tilde{w}_1^m \in  \mathcal{W}  : \tilde{w}_1^m \neq u_1^m+v_1^m,   \phi(\tilde{w}_1^m)+b_1+b_2 \in A_\epsilon^n(X_1+X_2|y) \}
\end{align*} 
Denote $w_1^m=u_1^m+v_1^m, Z=X_1+X_2$ and $W_i=U_i+V_i$ with probability one. Denote  $\mathcal{W}=\bigotimes_{i=1}^m A_\epsilon^{(k_i)}(W_i)$. By upper bounding the last indicator function, we get 

\begin{align*}
P(E | \phi, b_1, b_2) &\leq \sum_{u_1^m\in \mathcal{U}}\frac{1}{|\mathcal{U}|} \sum_{x_1\in H^n} \mathbbm{1}\{ x_1 = \phi(u_1^m)+b_1\} \sum_{v_1^m\in \mathcal{V}}\frac{1}{|\mathcal{V}|} \sum_{x_2\in H^n} \mathbbm{1}\{ x_2= \phi(v_1^m)+b_2\}\\
& \sum_{y\in A_\epsilon^n(Y|x_1,x_2)}  P(y|x_1, x_2) \sum_{\substack{ \tilde{w}_1^m\in \mathcal{W} \\ \tilde{w}_1^m \neq w_1^m}}\sum_{\tilde{z}\in A_\epsilon^n(Z|y)} \mathbbm{1}\{\tilde{z}=\phi(\tilde{w}_1^m)+b_1+b_2\}
\end{align*} 

Averaging over $\phi, b_1$ and $b_2$ yields:
\begin{align*}
P_e=\EE\{P(E|\Phi, B_1, B_2)\} \leq &\sum_{u_1^m\in \mathcal{U}}\frac{1}{|\mathcal{U}|} \sum_{x_1\in H^n} \sum_{v_1^m\in \mathcal{V}}\frac{1}{|\mathcal{V}|} \sum_{x_2\in H^n} \sum_{y\in A_\epsilon^n(Y|x_1,x_2)}  P(y|x_1, x_2) \sum_{\substack{ \tilde{w}_1^m\in \mathcal{W} \\ \tilde{w}_1^m \neq w_1^m}}\sum_{\tilde{z}\in A_\epsilon^n(Z|y)} \\ &P\{\tilde{z}=\Phi(\tilde{w}_1^m)+B_1+B_2, x_1 = \Phi(u_1^m)+B_1,x_2= \Phi(v_1^m)+B_2\}
\end{align*}

Notice that $B_1$ and $B_2$ are uniform over $\ZZ_{p^r}^n$ and independent of other random variables. Hence, the most inner term in the above summations is simplified to
\begin{align*}
p^{-2nr} P\{\tilde{z}-x-y=\Phi(\tilde{w}_1^m-w_1^m)\}
\end{align*}
Using  Lemma \ref{lem: P(phi)},  $P_e$ can be bounded as
\begin{align} \label{eq: pe_1}
P_e\leq &\sum_{u_1^m\in \mathcal{U}}\frac{1}{|\mathcal{U}|} \sum_{x_1\in \ZZ_{p^r}^n} \sum_{v_1^m\in \mathcal{V}}\frac{1}{|\mathcal{V}|} \sum_{x_2\in \ZZ_{p^r}^n} \sum_{y\in A_\epsilon^n(Y|x_1,x_2)}  P(y|x_1, x_2) \sum_{s=0}^{r-1}\sum_{\substack{ \tilde{w}_1^m\in \mathcal{W} \\ \tilde{w}_1^m-w_1^m \in H^k_{s}\backslash H^k_{s+1}}}\sum_{\substack{\tilde{z}\in A_\epsilon^n(Z|y)\\ \tilde{z}-x-y \in H^n_{s}}} p^{-2nr} p^{-n(r-s)}
\end{align}
Note the most inner term in the above summations does not depend on the value of $\tilde{z}$ and $\tilde{w}_1^m$.  Using Lemma \ref{lem: typical set intersection subgroup}, we have
\begin{align*}
P_e\leq &\sum_{u_1^m\in \mathcal{U}}\frac{1}{|\mathcal{U}|} \sum_{x_1\in \ZZ_{p^r}^n} \sum_{v_1^m\in \mathcal{V}}\frac{1}{|\mathcal{V}|} \sum_{x_2\in \ZZ_{p^r}^n}  \sum_{y\in A_\epsilon^n(Y|x_1,x_2)}  P(y|x_1, x_2) \sum_{s=0}^{r-1} 2^{kH(W|Q, [W]_s)}~ 2^{nH(Z|Y[Z]_{s})}p^{-2nr} p^{-n(r-s)},
\end{align*}
where $W=U+V$. As the terms in (\ref{eq: pe_1}) does not depend on the values of $\mathbf{u}_1^m, \mathbf{v}_1^m, \mathbf{x}_1, \mathbf{x}_2$ and  $\mathbf{y}$, we can replace the summations over them with the corresponding sets. As a result, we have 
\begin{align*}
P_e\leq &  \sum_{s=0}^{r-1}  p^{-n(r-s)}  2^{kH(W|Q, [W]_s)}~ 2^{nH(Z|Y[Z]_{s})}
\end{align*}
 Therefore, $P_e$ approaches zero as $n\rightarrow \infty$, if the following bounds hold:
\begin{align}\label{equ: bound simple form}
\frac{k}{n} H(W|Q, [W]_s) \leq   \log_2 p^{r-s}- H(Z|Y[Z]_{s}), \quad \mbox{for} ~ 0\leq s\leq r-1.
\end{align}
Since  $X_1+X_2$ is uniform over $\ZZ_{p^r}$, the right-hand side of (\ref{equ: bound simple form}) equals $I(Z;Y|[Z]_s)$. Note by definition $R_1=\frac{1}{n}\log_2 |\mathcal{C}_j|\leq \frac{1}{n}\log_2 |\mathcal{U}| \leq \frac{k}{n} H(U|Q)  $. Similarly, $R_2\leq \frac{k}{n} H(V|Q) $. Therefore, these inequalities along with (\ref{equ: bound simple form}) give,
\begin{align*}
R_1 \leq \frac{H(U|Q)}{H(W|Q, [W]_s) } I(X_1+X_2;Y|[X_1+X_2]_s)\\
R_2 \leq \frac{H(V|Q)}{H(W|Q, [W]_s) } I(X_1+X_2;Y|[X_1+X_2]_s).
\end{align*}
We take the union of the above region over all possible PMFs $p(u,v,q)=p(q)p(u|q)p(v|q)$, where $p(q)$ is a rational number. Since rational numbers are dense, we can let the PMF of $Q$ be arbitrary. the bounds in the theorem follows be denoting $W_1=U$ and $W_2=V$. Using the same argument as in Lemma \ref{lem: cardinality of Q}, we bound the cardinality of $Q$ by  $|\mathcal{Q}| \leq r$.


\begin{thebibliography}{1}
\bibitem{korner-marton}
J. Korner and K. Marton, ``How to encode the modulo-two sum of binary sources", IEEE Transactions on Information
Theory, IT-25:219–221, Mar. 1979.



\bibitem{Aron-BC-ISIT13}
A. Padakandla and S.S. Pradhan, ``Achievable rate region for three user discrete broadcast channel based on coset codes," IEEE International Symposium on Information Theory Proceedings (ISIT), pp.1277-1281, July 2013.






\bibitem{Philosof-Zamir}
T. Philosof and R. Zamir, ``On the loss of single-letter characterization: The dirty multiple access channel,"
IEEE Trans. Inform. Theory, vol. 55, pp. 2442-2454, June 2009.

\bibitem{Nazer_Gasper_Comp_MAC}
B. A. Nazer and M. Gastpar, ``Computation over multiple-access channels", IEEE Transactions on Information Theory, Oct. 2007.

\bibitem{Arun_comp_over_MAC_ISIT13}
A. Padakandla and S.S. Pradhan, ``Computing sum of sources over an arbitrary multiple access channel," IEEE International Symposium on Information Theory Proceedings (ISIT), 2013 , pp.2144-2148, July 2013.


\bibitem{Loeliger-91}
H. A. Loeliger, ``Signal sets matched to groups", IEEE Trans. Inform. Theory, vol. 37, no. 6, pp. 1675–1682, November
1991.

\bibitem{Loeliger-96}
H. A. Loeliger and T. Mittelholzer, ``Convolutional codes over groups”, IEEE Trans. Inform. Theory,vol. 42, no. 6, pp. 1660–1686, November 1996.

\bibitem{Como}
G. Como and F. Fagnani, ``The capacity of finite abelian group codes over symmetric memoryless channels", IEEE
Transactions on Information Theory, 55(5):2037–2054, 2009.


\bibitem{Aria_group_codes}
A. G. Sahebi, S.S. Pradhan, ``Abelian Group Codes for Channel Coding and Source Coding," IEEE Transactions on Information Theory, vol.61, no.5, pp.2399-2414, May 2015.



\bibitem{Aria_dist_source}
A. G. Sahebi, and S.S Pradhan, ``On distributed source coding using Abelian group codes," Communication, Control, and Computing (Allerton), 2012 50th Annual Allerton Conference on , pp.2068,2074, 1-5 Oct. 2012.


\bibitem{transversal}
M. Heidari, F. Shirani and S. Pradhan, `` Beyond group capacity in multi-terminal communications", IEEE International Symposium on Information Theory Proceedings (ISIT), pp. 2081-2085, July 2015.










\end{thebibliography}
\end{document}